\documentclass[a4paper,12pt]{article}%
\usepackage{amsfonts}
\usepackage{amssymb}
\usepackage{amsmath}
\usepackage{latexsym}
\usepackage{graphicx}
\usepackage[latin1]{inputenc}
\usepackage{hyperref}
\usepackage{amsthm}
\usepackage{graphicx}
\usepackage{color}%
\providecommand{\U}[1]{\protect\rule{.1in}{.1in}}
\font\teneusb=eusb10 \font\seveneusb=eusb7 \font\fiveeusb=eusb5
\newfam\eusbfam \textfont\eusbfam=\teneusb
\scriptfont\eusbfam=\seveneusb \scriptscriptfont\eusbfam=\fiveeusb

\font\tenbifull=cmmib10
\font\tenbimed=cmmib7
\font\tenbismall=cmmib5
\mathchardef\bbGamma="7000 \mathchardef\bbDelta="7001
\mathchardef\bbPhi="7002 \mathchardef\bbAlpha="7003
\mathchardef\bbXi="7004 \mathchardef\bbPi="7005
\mathchardef\bbSigma="7006 \mathchardef\bbUpsilon="7007
\mathchardef\bbTheta="7008 \mathchardef\bbPsi="7009
\mathchardef\bbOmega="700A \mathchardef\bbalpha="710B
\mathchardef\bbbeta="710C \mathchardef\bbgamma="710D
\mathchardef\bbdelta="710E \mathchardef\bbepsilon="710F
\mathchardef\bbzeta="7110 \mathchardef\bbeta="7111
\mathchardef\bbtheta="7112 \mathchardef\bbiota="7113
\mathchardef\bbkappa="7114 \mathchardef\bblambda="7115
\mathchardef\bbmu="7116 \mathchardef\bbnu="7117
\mathchardef\bbxi="7118 \mathchardef\bbpi="7119
\mathchardef\bbrho="711A \mathchardef\bbsigma="711B
\mathchardef\bbtau="711C \mathchardef\bbupsilon="711D
\mathchardef\bbphi="711E \mathchardef\bbchi="711F
\mathchardef\bbpsi="7120 \mathchardef\bbomega="7121
\mathchardef\bbvarepsilon="7122 \mathchardef\bbvartheta="7123
\mathchardef\bbvarpi="7124 \mathchardef\bbvarrho="7125
\mathchardef\bbvarsigma="7126 \mathchardef\bbvarphi="7127

\newcommand{\N}{{\rm I}\kern-0.18em{\rm N}}

\newcommand{\h}{{\rm I}\kern-0.18em{\rm H}}
\newcommand{\K}{{\rm I}\kern-0.18em{\rm K}}

\newcommand{\Z}{{\rm Z}\kern-0.34em{\rm Z}}

\newcommand{\1}{{\rm 1}\kern-0.22em{\rm I}}

\newtheorem{proposition}{Proposition}[section]
\newtheorem{cor}{Corollary}[section]
\newtheorem{theo}{Theorem}[section]
\newtheorem{lem}{Lemma}[section]
\newtheorem{rem}{Remark}[section]

\numberwithin{equation}{section}

\newcounter{eqroman}
\begin{document}

\title{Weighted sampling, Maximum Likelihood and minimum divergence estimators}
\author{Michel Broniatowski$^{1}$, Zhansheng Cao$^{1}$\\$^{1}$LSTA, Université Paris 6, France}
\maketitle

\begin{abstract}
This paper explores Maximum Likelihood in parametric models in the context of
Sanov type Large Deviation Probabilities. MLE in parametric models under
weighted sampling is shown to be\ associated with the minimization of a
specific divergence criterion defined with respect to the distribution of the
weights.\ Some properties of the resulting inferential procedure are
presented; Bahadur efficiency of tests are also considered in this context.

\end{abstract}

\section{Motivation and context}

This paper explores Maximum Likelihood paradigm in the context of sampling. It
mainly quotes that inference criterion is strongly connected with the sampling
scheme generating the data. Under a given model, when i.i.d. sampling is
considered and some standard regularity is assumed, then the Maximum
Likelihood principle\ loosely states that conditionally upon the observed
data, resampling under the same i.i.d. scheme should resemble closely to the
initial sample only when the resampling distribution is close to the initial
unknown one.

Keeping the same definition it appears that under other sampling schemes, the
Maximum Likelihood Principle yields a wide range of statistical procedures.
Those have in common with the classical simple i.i.d. sampling case that they
can be embedded in a natural class of methods based on minimization of $\phi
-$divergences between the empirical measure of the data and the model. In the
classical i.i.d. case the divergence is the Kullback-Leibler one, which yields
the standard form of the Likelihood function. In the case of the weighted
bootstrap, the divergence to be optimized is directly related to the
distribution of the weights.

This paper discusses the choice of an inference criterion in parametric
setting.\ We consider a wide range of commonly used statistical criterions,
namely all those induced by the so-called power divergence, including
therefore Maximum Likelihood, Kullback-Leibler, Chi-square, Hellinger
distance, etc.\ The steps of the discussion are as follows.

We first insert Maximum Likelihood paradigm at the center of the scene,
putting forwards its strong connection with large deviation probabilities for
the empirical measure. The argument can be sketched as follows: for any
putative $\theta$ in the parameter set, consider $n$ virtual simulated r.v's
$X_{i,\theta}$ with corresponding empirical measure $P_{n,\theta}.$ Evaluate
the probability that $P_{n,\theta}$ is close to $P_{n}$ , conditionally on
$P_{n}$, the empirical measure pertaining to the observed data; such statement
is refered to as a conditional Sanov theorem, and for any $\theta$ this
probability is governed by the Kullback-Leibler distance between $P_{\theta}$
and $P_{\theta_{T}}$ where $\theta_{T}$ stands for the true value of the
parameter. Estimate this probability for any $\theta$, obviously based on the
observed data. Optimize in $\theta;$ this provides the MLE, as shown in the
two cases of the i.i.d. sample scheme; our first example is the case when the
observations take values in a finite set, and the second case (infinite case),
helps to set the arguments to be put forwards. Introducing MLE's through Large
deviations for the empirical measure is in the vein of various recent
approaches; see Grendar and Judge $\cite{Judge}$.

We next consider a generalized sampling scheme inherited from the bootstrap,
which we call weighted sampling; it amounts to introduce a family of i.i.d.
weights $W_{1},...,W_{n}$ with mean and variance $1.$ \ The corresponding
empirical measure pertaining to the data set $x_{1}$ $,..,x_{n}$ is just the
weighted empirical measure.\ The MLE is defined through a similar procedure as
just evoqued. The conditional Sanov Theorem is governed by a divergence
criterion which is defined through the distribution of the weights. Hence MLE
results in the optimization of a divergence measure between distributions in
the model and the weighted empirical measure pertaining to the dataset.

Resulting properties of the estimators are studied.

Optimization of $\phi-$divergences between the empirical measure of the data
and the model is problematic when the support of the model is not finite. A
number of authors have considered so-called dual representation formulas for
divergences or, globally, for convex pseudodistances between
distributions.\ We will make use of the one exposed in \cite{BK2009}; see also
\cite{BrArxiv} for an easy derivation.

\subsection{Notation}

\subsubsection{\bigskip Divergences}

The space $S$ is a Polish space endowed with its Borel field $\mathcal{B}%
\left(  S\right)  .$ We consider an identifiable parametric model
$\mathcal{P}_{\Theta}$ on $\left(  S,\mathcal{B}\left(  S\right)  \right)  $,
hence a class of probability distributions $P_{\theta}$ indexed by a subset
$\Theta$ included in $\mathbb{R}^{d}$; $\Theta$ needs not be open$.$ The class
of all probability measures on $\left(  S,\mathcal{B}\left(  S\right)
\right)  $ is denoted $\mathcal{P}$ and $\mathcal{M}(S)$ designates the class
of all finite signed measures on $\left(  S,\mathcal{B}\left(  S\right)
\right)  .$

A non negative convex function $\varphi$ with values in $\overline
{\mathbb{R}^{+}}$ belonging to $C^{2}$ $\left(  \mathbb{R}\right)  $ and
satisfying $\varphi\left(  1\right)  =\varphi^{\prime}\left(  1\right)  =0$
and $\varphi^{\prime\prime}\left(  1\right)  $ is a \textit{divergence
function.} An important class of such functions is defined through the power
divergence functions%
\begin{equation}
\varphi_{\gamma}\left(  x\right)  :=\frac{x^{\gamma}-\gamma x+\gamma-1}%
{\gamma\left(  \gamma-1\right)  } \label{CRdiv}%
\end{equation}
defined for all real $\gamma\neq0,1$ with $\varphi_{0}\left(  x\right)
:=-\log x+x-1$ (the likelihood divergence function) and $\varphi_{1}\left(
x\right)  :=x\log x-x+1$ (the Kullback-Leibler divergence function). This
class is usually refered to as the Cressie-Read family of divergence
functions, a custom we will follow, although its origin takes from
\cite{Renyi61}. When $x$ is such that $\varphi_{\gamma}\left(  x\right)  $ is
undefined by the above definitions, we set $\varphi_{\gamma}\left(  x\right)
:=+\infty$, by which the definition above is satisfied for all $\varphi
_{\gamma}$. It consists in the simplest power-type class of functions (with
the limits in $\gamma\rightarrow0,1$) which fulfill the definition. The
$L_{1}$ divergence function $\varphi\left(  x\right)  :=\left\vert
x-1\right\vert $ is not captured by the Cressie-Read family of functions.

Associated with a divergence function $\varphi$ is the \textit{divergence
pseudodistance} between a probability measure and a finite signed measure; see
\cite{BrVajda}.

For $P$ and $Q$ in $\mathcal{M}$ define%
\begin{align*}
\phi\left(  Q,P\right)   &  :=\int\varphi\left(  \frac{dQ}{dP}\right)
dP\text{ \ whenever }Q\text{ is a.c. w.r.t. }P\\
&  :=+\infty\text{ \ \ otherwise.}%
\end{align*}
The divergence $\phi\left(  Q,P\right)  $ is best seen as a mapping
$Q\rightarrow\phi\left(  Q,P\right)  $ from $\mathcal{M}$ onto $\overline
{\mathbb{R}^{+}}$ for fixed $P$ in $\mathcal{M}$. Indexing this pseudodistance
by $\gamma$ and using $\varphi_{\gamma}$ as divergence function yields the
likelihood divergence $\phi_{0}\left(  Q,P\right)  :=-\int\log\left(
\frac{dQ}{dP}\right)  dP$, the Kullback-Leibler divergence $\phi_{1}\left(
Q,P\right)  :=\int\log\left(  \frac{dQ}{dP}\right)  dQ$, the Hellinger
divergence $\phi_{1/2}\left(  Q,P\right)  :=\frac{1}{2}\int\left(  \sqrt
{\frac{dQ}{dP}-1}\right)  ^{2}dP$, the modified $\chi^{2}$ divergence
$\phi_{-1}\left(  Q,P\right)  :=\frac{1}{2}\int\left(  \frac{dQ}{dP}-1\right)
^{2}\left(  \frac{dQ}{dP}\right)  ^{-1}dP$. All these divergences are defined
on $\mathcal{P}$. The $\chi^{2}$ divergence $\phi_{2}\left(  Q,P\right)
:=\frac{1}{2}\int\left(  \frac{dQ}{dP}-1\right)  ^{2}dP$ is defined on
$\mathcal{M}$. We refer to \cite{BK2009} for the advantage to extend the
definition to possibly signed measures in the context of parametric inference
for non regular models.

The conjugate divergence function of $\varphi$ is defined through%
\begin{equation}
\widetilde{\varphi}\left(  x\right)  :=x\varphi\left(  \frac{1}{x}\right)
\label{conjugatediv}%
\end{equation}
and the corresponding divergence pseudodistance $\widetilde{\phi}\left(
P,Q\right)  $ is
\[
\widetilde{\phi}\left(  P,Q\right)  :=\int\widetilde{\varphi}\left(  \frac
{dP}{dQ}\right)  dQ
\]
which satisfies%
\[
\widetilde{\phi}\left(  P,Q\right)  =\phi\left(  Q,P\right)
\]
whenever defined, and equals $+\infty$ otherwise. When $\varphi=\varphi
_{\gamma}$ then $\widetilde{\varphi}=\varphi_{1-\gamma}$ as follows by
substitution. Pairs $\left(  \varphi_{\gamma},\varphi_{1-\gamma}\right)  $ are
therefore \textit{conjugate pairs}. Inside the Cressie-Read family, the
Hellinger divergence function is self-conjugate.

In parametric models $\varphi-$divergences between two distributions take a
simple variational form.\ It holds, when $\varphi$ is a differentiable
function, and under a commonly met regularity condition, denoted (RC) in
\cite{BrArxiv}
\begin{equation}
\phi(P_{\theta},P_{\theta_{T}})=\sup_{\alpha\in\mathcal{U}}\int\varphi
^{\prime}\left(  \frac{dP_{\theta}}{dP_{\alpha}}\right)  dP_{\theta}%
-\int\varphi^{\#}\left(  \frac{dP_{\theta}}{dP_{\alpha}}\right)
dP_{\theta_{T}} \label{dual param}%
\end{equation}
where $\varphi^{\#}(x):=x\varphi^{\prime}(x)-\varphi(x).$ In the above
formula, $\mathcal{U}$ designates a subset of $\Theta$ containing $\theta_{T}$
such that for any $\theta,\theta^{\prime}$ in $\mathcal{U}$, $\phi\left(
P_{\theta},P_{\theta^{\prime}}\right)  $ is finite. This formula holds for any
divergence in the Cressie Read family, as considered here.

Denote
\[
h(\theta,\alpha,x):=\int\varphi^{\prime}\left(  \frac{dP_{\theta}}{dP_{\alpha
}}\right)  dP_{\theta}-\varphi^{\#}\left(  \frac{dP_{\theta}}{dP_{\alpha}%
}\left(  x\right)  \right)
\]
from which
\begin{equation}
\phi(P_{\theta},P_{\theta_{T}}):=\sup_{\alpha\in\mathcal{U}}\int
h(\theta,\alpha,x)dP_{\theta_{T}}(x). \label{dual h}%
\end{equation}
For CR divergences
\[
h(\theta,\alpha,x)=\frac{1}{\gamma-1}\left[  \int\left(  \frac{dP_{\theta}%
}{dP_{\alpha}}\right)  ^{\gamma-1}dP_{\theta}-1\right]  -\frac{1}{\gamma
}\left[  \left(  \frac{dP_{\theta}}{dP_{\alpha}}\left(  x\right)  \right)
^{\gamma}-1\right]  .
\]

\subsubsection{Weights}

For a given real valued random variable $W$ denote%
\begin{equation}
M(t):=\log E\exp tW \label{mgfW}%
\end{equation}
its cumulant generating function which we assume to be finite in a non void
interval including $0$ (this is the so-called Cramer condition)$.$ The Fenchel
Legendre transform of $M$ is also called the Chernoff function and is defined
through
\begin{equation}
\varphi^{W}(x)=M^{\ast}(x):=\sup_{t}tx-M(t). \label{phiWduale de MgfW}%
\end{equation}
The function $x\rightarrow\varphi^{W}(x)$ is non negative, is $C^{2}$ and
convex. We also assume that $EW=1$ together with $VarW=1$ which implies
$\varphi^{W}(1)=\left(  \varphi^{W}\right)  ^{\prime}(1)=0$ and $\left(
\varphi^{W}\right)  ^{\prime\prime}(1)=1.$ Hence $\varphi^{W}(x)$ is a
divergence function with corresponding divergence pseudodistance $\phi^{W}$ .
Associated with $\varphi^{W}$ is the conjugate divergence $\widetilde{\phi
^{W}}$ with divergence function $\widetilde{\varphi^{W}}$ , which therefore
satisfies%
\[
\phi^{W}\left(  Q,P\right)  =\widetilde{\phi^{W}}\left(  P,Q\right)  .
\]

\subsubsection{Measure spaces}

\label{measureSpace}

This paper makes extensive use of Sanov type large deviation results for
empirical measures or weighted empirical measures. This requires some
definitions and facts.

The vector space $\mathcal{M}(S)$ is endowed with the $\tau-$topology, which
is the coarest making all mappings $Q\rightarrow\int fdQ$ continuous for any
$Q\in\mathcal{M}(S)$ and any $f\in B(S)$ which denotes the class of all
bounded measurable functions on $\left(  S,\mathcal{B}\left(  S\right)
\right)  .$ A slightly stronger topology will be used in this paper, the
$\tau_{0}$ topology, introduced in \cite{CS1984}, which is the natural
setting for our sake. This topology can be described through the following
basis of neighborhoods. Consider $\mathfrak{P}$ the class of all partitions of
$S$ and for $k\geq1$ the class $\mathfrak{P}_{k}$ of all partitions of $S$
into $k$ disjoint sets, $\mathcal{P}_{k}:=\left(  A_{1},...,A_{k}\right)  $
where the $A_{i}$'s belong to $\mathcal{B}\left(  S\right)  .$ For fixed $P$
in $\mathcal{M}$, for any $k$, any such partition $\mathcal{P}_{k}$ in
$\mathfrak{P}_{k}$ and any positive $\varepsilon$ define the open neighborhood
$U\left(  P,\varepsilon,\mathcal{P}_{k}\right)  $ through%
\[
U\left(  P,\varepsilon,\mathcal{P}_{k}\right)  :=\left\{  Q\in\mathcal{M}%
\text{ such that }\max_{1\leq i\leq k}\left\vert P(A_{i})-Q(A_{i})\right\vert
<\varepsilon\text{ and }Q(A_{i})=0\text{ if }P(A_{i})=0\right\}  .
\]
The additional requirement $Q(A_{i})=0$ if $P(A_{i})=0$ in the above
definition with respect to the classical definition of the basis of the
$\tau-$topology is essential for the derivation of Sanov type
theorems.\ Endowed with the $\tau_{0}-$topology, $\mathcal{M}$ is a Hausdorff
locally convex vector space.

The following Pinsker type property holds%
\[
\sup_{k}\sum_{i=1}^{k}\varphi\left(  \frac{Q\left(  A_{i}\right)  }{P\left(
A_{i}\right)  }\right)  P\left(  A_{i}\right)  =\phi\left(  Q,P\right)
\]
see \cite{HoerenmoesVajda}.

For any $P$ in $\mathcal{M}$ the mapping $Q\rightarrow\phi(Q,P)$ is lower semi
continuous; see \cite{BK2006}, Proposition $2.2$. Denoting $\left(
a,b\right)  $ the domain of $\varphi$ whenever
\[
\lim_{\substack{x\rightarrow a\\x>a}}\frac{\varphi(x)}{x}=\lim
_{\substack{x\rightarrow b\\x<b}}\frac{\varphi(x)}{x}=+\infty
\]
then for any positive $C,$ the level set $\left\{  Q:\phi\left(  Q,P\right)
\leq C\right\}  $ is $\tau_{0}-$compact, making $Q\rightarrow\phi(Q,P)$ a
so-called good rate function. \ Divergence functions $\varphi$ satisfying this
requirement for example are $\varphi_{\gamma}$ with $\gamma>1;$ see
\cite{BK2006} for different cases.

\subsubsection{Minimum dual divergence estimators}

The above formula (\ref{dual param}) defines a whole range of plug in
estimators of $\phi(P_{\theta},P_{\theta_{T}})$ and of $\theta_{T}.$ Let
$X_{1},...,X_{n}$ denote $n$ i.i.d. r.v's with common didistribution
$P_{\theta_{T}}.$ Denoting
\[
P_{n}:=\frac{1}{n}\sum_{i=1}^{n}\delta_{X_{i}}%
\]
the empirical measure pertaining to this sample. The plug in estimator of
$\phi(P_{\theta},P_{\theta_{T}})$ is defined through
\[
\phi_{n}(P_{\theta},P_{\theta_{T}}):=\sup_{\alpha\in\mathcal{U}}\int
h(\theta,\alpha,x)dP_{n}\left(  x\right)
\]
and the family of M-estimators indexed by $\theta$
\[
\alpha_{n}\left(  \theta\right)  :=\arg\sup_{\alpha\in\mathcal{U}}\int
h(\theta,\alpha,x)dP_{n}\left(  x\right)
\]
approximates $\theta_{T}$. In the above formulas $\mathcal{U}$ is defined
after (\ref{dual param}). See \cite{BK2009} and \cite{TomaBroniatowski} for
asymptotic properties and robustness results.\

Since $\phi(P_{\theta_{T}},P_{\theta_{T}})=0$ a natural estimator of
$\theta_{T}$ which only depends on the choice of the divergence function
$\varphi$ is defined through
\begin{align*}
\theta_{n}  &  :=\arg\inf_{\theta}\phi_{n}(P_{\theta},P_{\theta_{T}})\\
&  =\arg\inf_{\theta\in\mathcal{U}}\sup_{\alpha\in\mathcal{U}}\int
h(\theta,\alpha,x)dP_{n}\left(  x\right)  ;
\end{align*}
see \cite{BK2009} for limit properties.

\bigskip

\section{Large deviation and maximum likelihood}

\subsection{Maximum likelihood under finite supported distributions and simple
sampling}

Suppose that all probability measures $P_{\theta}$ in $\mathcal{P}_{\Theta} $
share the same finite support $S:=\left\{  1,...,k\right\}  .$ Let
$X_{1},...X_{n}$ be a set of $n$ independent random variables with common
probability measure $P_{\theta_{T}}$ and consider the Maximum Likelihood
estimator of $\theta_{T}$ . A common way to define the ML paradigm is as
follows: For any $\theta$ consider independent random variables $\left(
X_{1,\theta},...X_{n,\theta}\right)  $ with probability measure $P_{\theta} $
, thus \textit{sampled in the same way as the }$X_{i}$\textit{'s}, but under
some altermative $\theta.$ Define $\theta_{ML}$ as the value of the parameter
$\theta$ for which the probability that, up to a permutation of the order of
the $X_{i,\theta}$'s, the probability that $\left(  X_{1,\theta}%
,...X_{n,\theta}\right)  $ occupies $S$ as does $X_{1},...X_{n}$ is maximal,
conditionaly on the observed sample $X_{1},...X_{n}.$ In formula, let $\sigma$
denote a random permutation of the indexes $\left\{  1,2,...,n\right\}  $ and
$\theta_{ML}$ is defined through
\begin{equation}
\theta_{ML}:=\arg\max_{\theta}\frac{1}{n!}\sum_{\sigma\in\mathfrak{S}%
}P_{\theta}\left(  \left.  \left(  X_{\sigma(1),\theta},...,X_{\sigma
(n),\theta}\right)  =\left(  X_{1},...X_{n}\right)  \right\vert \left(
X_{1},...X_{n}\right)  \right)  \label{MLfinite}%
\end{equation}
where the summation is extended on all equally probable permutations of
$\left\{  1,2,...,n\right\}  .$

Denote%

\[
P_{n}:=\frac{1}{n}\sum_{i=1}^{n}\delta_{X_{i}}%
\]
and
\[
P_{n,\theta}:=\frac{1}{n}\sum_{i=1}^{n}\delta_{X_{i,\theta}}%
\]
the empirical measures pertaining respectively to $\left(  X_{1}%
,...X_{n}\right)  $ and $\left(  X_{1,\theta},...X_{n,\theta}\right)  $

An alternative expression for $\theta_{ML}$ is
\begin{equation}
\theta_{ML}:=\arg\max_{\theta}P_{\theta}\left(  \left.  P_{n,\theta}%
=P_{n}\right\vert P_{n}\right)  . \label{MLfinitebis}%
\end{equation}

\bigskip An explicit enumeration of the above expression $P_{\theta}\left(
\left.  P_{n,\theta}=P_{n}\right\vert P_{n}\right)  $ involves the quantities
\[
n_{j}:=card\left\{  i:X_{i}=j\right\}
\]
for $j=1,...,k$ and yields%
\begin{equation}
P_{\theta}\left(  \left.  P_{n,\theta}=P_{n}\right\vert P_{n}\right)  =\frac{%
%TCIMACRO{\dprod \limits_{j=1}^{k}}%
%BeginExpansion
{\displaystyle\prod\limits_{j=1}^{k}}
%EndExpansion
n_{j}!P_{\theta}\left(  j\right)  ^{n_{j}}}{n!} \label{multinomial}%
\end{equation}
as follows from the classical multinomial distribution. Optimizing on $\theta$
in (\ref{multinomial}) yields
\begin{align*}
\theta_{ML}  &  =\arg\max_{\theta}\sum_{j=1}^{k}\frac{n_{j}}{n}\log P_{\theta
}\left(  j\right) \\
&  =\arg\max_{\theta}\frac{1}{n}\sum_{i=1}^{n}\log P_{\theta}\left(
X_{i}\right)  .
\end{align*}
Consider now the Kullback-Leibler distance between $P_{\theta}$ and $P_{n}$
which is non commutative and defined through
\begin{align}
KL\left(  P_{n},P_{\theta}\right)   &  :=\sum_{j=1}^{k}\varphi\left(
\frac{n_{j}/n}{P_{\theta}\left(  j\right)  }\right)  P_{\theta}\left(
j\right) \nonumber\\
&  =\sum_{j=1}^{k}\left(  n_{j}/n\right)  \log\frac{n_{j}/n}{P_{\theta}\left(
j\right)  } \label{KLempfini}%
\end{align}
where
\begin{equation}
\varphi(x):=x\log x-x+1 \label{divKL}%
\end{equation}
which is the Kullback-Leibler divergence function. Minimizing the
Kullback-Leibler distance $KL\left(  P_{n},P_{\theta}\right)  $ upon $\theta$
yields
\begin{align*}
\theta_{KL}  &  =\arg\min_{\theta}KL\left(  P_{n},P_{\theta}\right) \\
&  =\arg\min_{\theta}-\sum_{j=1}^{k}\frac{n_{j}}{n}\log P_{\theta}\left(
j\right) \\
&  =\arg\max_{\theta}\sum_{j=1}^{k}\frac{n_{j}}{n}\log P_{\theta}\left(
j\right) \\
&  =\theta_{ML}.
\end{align*}
Introduce the \textit{conjugate divergence function }$\widetilde{\varphi}$
$\ $of $\varphi$ , inducing the modified Kullback-Leibler, or so-called
Likelihood divergence pseudodistance $KL_{m}$ which therefore satisfies
\[
KL_{m}\left(  P_{\theta},P_{n}\right)  =KL\left(  P_{n},P_{\theta}\right)  .
\]
We have proved that minimizing the Kullback-Leibler divergence $KL\left(
P_{n},P_{\theta}\right)  $ amounts to minimizing the Likelihood divergence
$KL_{m}\left(  P_{\theta},P_{n}\right)  $ and produces the ML estimate of
$\theta_{T}.$

Kullback-Leibler divergence as defined above by $KL\left(  P_{n},P_{\theta
}\right)  $ is related to the way $P_{n}$ keeps away from $P_{\theta}$ when
$\theta$ is not equal to the true value of the parameter $\theta_{T}$
generating the observations $X_{i}$'s and is closely related with the type of
sampling of the $X_{i}$'s. In the present case i.i.d. sampling of the
$X_{i,\theta}$'s under $P_{\theta}$ results in the asymptotic property, named
Large Deviation Sanov property
\begin{equation}
\lim_{n\rightarrow\infty}\frac{1}{n}\log P_{\theta}\left(  \left.
P_{n,\theta}=P_{n}\right\vert P_{n}\right)  =-KL\left(  P_{\theta_{T}%
},P_{\theta}\right)  . \label{Sanov fini}%
\end{equation}
This result can easily be obtained from (\ref{multinomial}) using Stirling
formula to handle the factorial terms and the law of large numbers which
states that for all $j$'s, $n_{j}/n$ tends to $P_{\theta_{T}}(j)$ as $n$ tends
to infinity. Comparing with (\ref{KLempfini}) we note that the ML estimator
$\theta_{ML}$ estimates the minimizer of the natural estimator of $KL\left(
P_{\theta_{T}},P_{\theta}\right)  $ in $\theta,$ substituting the unknown
measure generating the $X_{i}$'s by its empirical counterpart $P_{n}$ .
Alternatively as will be used in the sequel, $\theta_{ML}$ minimizes upon
$\theta$ the Likelihood divergence $KL_{m}\left(  P_{\theta},P_{\theta_{T}%
}\right)  $ between $P_{\theta}$ and $P_{\theta_{T}}$ substituting the unknown
measure $P_{\theta_{T}}$ generating the $X_{i}$'s by its empirical counterpart
$P_{n}$ . Summarizing we have obtained:

The ML\ estimate can be obtained from a LDP\ statement as given in
(\ref{Sanov fini}), optimizing in $\theta$ in the estimator of the LDP rate
where the plug-in method of the empirical measure of the data is used instead
of the unknown measure $P_{\theta_{T}}.$ Alternatively it holds
\begin{equation}
\theta_{ML}:=\arg\min_{\theta}\widehat{KL_{m}}\left(  P_{\theta},P_{\theta
_{T}}\right)  \label{ML finite case}%
\end{equation}
with
\[
\widehat{KL_{m}}\left(  P_{\theta},P_{\theta_{T}}\right)  :=KL_{m}\left(
P_{\theta},P_{n}\right)  .
\]
In the rest of this section we will develop a similar approach for a model
$\mathcal{P}_{\Theta}$ whose all members $P_{\theta}$ share the same infinite
(countable or not) support $S.$

The statistical properties of $\theta_{ML}$ are obtained under the i.i.d.
sampling having generated the observed values.

This principle will be kept throughout this paper: the estimator is defined as
maximizing the probability that the simulated empirical measure be close to
the empirical measure as observed on the sample, conditionally on it,
following the same sampling scheme. This yilds a maximum likelihood estimator,
and its properties a re then obtained when randomness is introduced as
resulting from the sampling scheme.

\subsection{Maximum likelihood under general distributions and simple
sampling}

When the support of the generic r.v. $X_{1}$ is not finite some of the
arguments above are not valid any longer and some discretization scheme is
required in order to get occupation probabilities in the spirit of
(\ref{multinomial}) or (\ref{Sanov fini}). Since all distributions $P_{\theta
}$ in $\mathcal{P}_{\Theta}$ have infinite support, i.i.d. sampling under any
$P_{\theta}$ yields $\left(  X_{1,\theta},...X_{n,\theta}\right)  $ such that
\[
P_{\theta}\left(  \left.  P_{n,\theta}=P_{n}\right\vert P_{n}\right)  =0
\]
for all $n$, so that we are lead to consider the optimization upon $\theta$ of
probabilities of the type $P_{\theta}\left(  \left.  P_{n,\theta}\in V\left(
P_{n}\right)  \right\vert P_{n}\right)  $ where $V\left(  P_{n}\right)  $ is a
(small) neighborhood of $P_{n}.$ Considering the distribution of the outcomes
of the simulating scheme $P_{\theta}$ results in the definition of
neighborhoods through partitions of $S$, hence through the $\tau_{0}-$topology.

When $P_{n}$ is the empirical measure for some observed r.v's $X_{1},...X_{n}
$ , an $\varepsilon-$neighborhood of $P_{n}$ \ contains distributions whose
support is not necessarily finite, and may indeed be equivalent to the
measures in the model $\mathcal{P}_{\Theta}$ when defined on the Borel
$\sigma-$field $\mathcal{B}\left(  S\right)  $.\

Let $\mathcal{P}_{k}:=\left(  A_{1},...,A_{k}\right)  $ be some partition in
$\mathfrak{P}_{k}.$ Denote
\begin{equation}
V_{k,\varepsilon}\left(  P_{n}\right)  :=\left\{  Q\in\mathcal{M}\text{ such
that }\max_{i=1,...,k}\left\vert P_{n}(A_{i})-Q(A_{i})\right\vert
<\varepsilon\text{ and }Q(A_{i})=0\text{ if }P_{n}(A_{i})=0\right\}
\label{voisinageV}%
\end{equation}
an open neighborhood of $P_{n}.$

We also would define the Kullback-Leibler divergence between two probability
measures $Q$ and $P$ on the partition $\mathcal{P}_{k}$ through
\[
KL_{A_{k}}\left(  Q,P\right)  :=\sum_{A_{j}\in\mathcal{P}_{k}}\log\left(
\frac{Q(A_{j})}{P(A_{j})}\right)  Q(A_{j}).
\]
Also we define the corresponding Likelihood divergence on $\mathcal{P}_{k}$
through
\begin{equation}
\left(  KL_{m}\right)  _{\mathcal{P}_{k}}\left(  Q,P\right)  :=KL_{\mathcal{P}%
_{k}}\left(  P,Q\right)  . \label{divML}%
\end{equation}

As in the finite case for any $\theta$ in $\Theta$ denote $\left(
X_{1,\theta},...X_{n,\theta}\right)  $ a set of $n$ i.i.d.\ random variables
with common distribution $P_{\theta}.$ We have

\begin{lem}
\label{Lemma inf}For large $n$
\begin{align*}
\frac{1}{n}\log P_{\theta}\left(  \left.  P_{n,\theta}\in V_{k,\varepsilon
}\left(  P_{n}\right)  \right\vert P_{n}\right)   &  \geq-KL_{\mathcal{P}_{k}%
}\left(  V_{k,\varepsilon}\left(  P_{n}\right)  ,P_{\theta}\right)
-\frac{k\log(n+1)}{n}\\
&  :=-\inf_{Q\in V_{k,\varepsilon}\left(  P_{n}\right)  }KL_{\mathcal{P}_{k}%
}\left(  Q,P_{\theta}\right)  -\frac{k\log(n+1)}{n}%
\end{align*}

\end{lem}

\begin{proof}The proof uses similar arguments as in \cite{CS1984} Lemma 4.1.\ For fixed
$k$ and large $n,$ $P_{\theta_{T}}$ belongs to $V_{k,\varepsilon}\left(
P_{n}\right)  $, by the law of large numbers. Indeed for large $n$ ,
$P_{n}\left(  A_{j}\right)  $ is positive and $\left\vert P_{\theta_{T}%
}\left(  A_{j}\right)  -P_{n}\left(  A_{j}\right)  \right\vert <\varepsilon$
for all $j$ in $\left\{  1,...,k\right\}  .$ Assuming that for all $\theta$ in
$\Theta$
\[
KL\left(  P_{\theta_{T}},P_{\theta}\right)  <\infty
\]
and taking into account the fact\ (see \cite{Pinsker64}) that for any
probability measures $P$ and $Q$, $K(P,Q)=\sup_{k}\sup_{\mathcal{P}_{k}%
\in\mathfrak{P}_{k}}KL_{\mathcal{P}_{k}}\left(  P,Q\right)  $ where
$\mathfrak{P}_{k}$ is the class of all partitions of $S$ in $k$ sets in
$\mathcal{B}\left(  S\right)  $, it follows that
\[
KL_{\mathcal{P}_{k}}\left(  V_{k,\varepsilon}\left(  P_{n}\right)  ,P_{\theta
}\right)  \text{ is finite}%
\]
for all fixed $k$ and large $n.$ For positive $\delta$ let $P^{(n)}$ in
$V_{k,\varepsilon}\left(  P_{n}\right)  $ with
\[
KL_{\mathcal{P}_{k}}\left(  P^{(n)},P_{\theta}\right)  <KL_{\mathcal{P}_{k}%
}\left(  V_{k,\varepsilon}\left(  P_{n}\right)  ,P_{\theta}\right)  +\delta.
\]
Let $0<$ $\varepsilon^{\prime}<\varepsilon$ and non negative numbers $r_{j}$ ,
$1\leq j\leq k$ such that
\[
\left\vert r_{j}-P^{(n)}\left(  A_{j}\right)  \right\vert <\varepsilon
^{\prime}\text{, and }r_{j}=0\text{ if }P^{(n)}\left(  A_{j}\right)  =0\text{
and }\sum_{j=1}^{k}r_{j}=1.
\]
The probability vector $\left(  r_{1},...,r_{k}\right)  $ defines a
probability measure $R$ on $\left(  S,\mathcal{P}_{k}\right)  ,$ and $R$
belongs to $V_{k,\varepsilon}\left(  P_{n}\right)  .$ By continuity of the
mapping $x\rightarrow x\log\frac{x}{P_{\theta}\left(  A_{j}\right)  }$ it is
possible to fit the $r_{j}$'s such that for all $j$ between $1$ and $k$
\begin{equation}
\left\vert r_{j}\log\frac{r_{j}}{P_{\theta}\left(  A_{j}\right)  }%
-P^{(n)}\left(  A_{j}\right)  \log\frac{P^{(n)}\left(  A_{j}\right)
}{P_{\theta}\left(  A_{j}\right)  }\right\vert <\frac{\delta}{k}.
\label{condR_j}%
\end{equation}
Indeed since all the $P_{\theta}$'s share the same support, if $P_{\theta
}\left(  A_{j}\right)  =0$ then $P_{\theta_{T}}\left(  A_{j}\right)  =0$ which
in turn yields $P_{n}(A_{j})=0$ which through (\ref{voisinageV}) implies
$P^{(n)}\left(  A_{j}\right)  =0.$ This plus the conventions $0/0=0$ and
$0\log0=0$ implies that (\ref{condR_j}) holds true for some choice of the
$r_{j}$'s. Choose further the $r_{j}$'s in such a way that $l_{j}:=nr_{j}$ is
an integer for all $j.$ Let $P_{n,\theta}$ denote the empirical distribution
of the $X_{i,\theta}$'s. We now proceed to the evaluation of $P_{\theta
}\left(  \left.  P_{n,\theta}\in V_{k,\varepsilon}\left(  P_{n}\right)
\right\vert P_{n}\right)  .$ It holds%
\begin{align*}
P_{\theta}\left(  \left.  P_{n,\theta}\in V_{k,\varepsilon}\left(
P_{n}\right)  \right\vert P_{n}\right)   &  \geq P_{\theta}\left(  \left.
P_{n,\theta}\left(  A_{j}\right)  =r_{j},1\leq j\leq k\right\vert P_{n}\right)
\\
&  =\frac{%
%TCIMACRO{\dprod \limits_{j=1}^{k}}%
%BeginExpansion
{\displaystyle\prod\limits_{j=1}^{k}}
%EndExpansion
l_{j}!}{n!}%
%TCIMACRO{\dprod \limits_{j=1}^{k}}%
%BeginExpansion
{\displaystyle\prod\limits_{j=1}^{k}}
%EndExpansion
P_{\theta}\left(  A_{j}\right)  ^{l_{j}}\\
&  \geq\left(  n+1\right)  ^{-k}\exp-n\sum_{j=1}^{k}r_{j}\log\frac{r_{j}%
}{P_{\theta}\left(  A_{j}\right)  }%
\end{align*}
where we used the same argument as in \cite{CS1984}, Lemma 4.1. In turn
using (\ref{condR_j})
\begin{align*}
\sum_{j=1}^{k}r_{j}\log\frac{r_{j}}{P_{\theta}\left(  A_{j}\right)  }  &
\leq\sum_{j=1}^{k}P^{(n)}\left(  A_{j}\right)  \log\frac{P^{(n)}\left(
A_{j}\right)  }{P_{\theta}\left(  A_{j}\right)  }+\delta\\
&  \leq KL_{\mathcal{P}_{k}}\left(  V_{k,\varepsilon}\left(  P_{n}\right)
,P_{\theta}\right)  +2\delta
\end{align*}
and the proof is completed.
\end{proof}

The reverse inequality is as in \cite{CS1984} p 790: The set
$V_{k,\varepsilon}\left(  P_{n}\right)  $ is completely convex, in the
terminology of \cite{CS1984}, whence it follows

\begin{lem}
\label{Lemma sup}For all $n$%
\[
\frac{1}{n}\log P_{\theta}\left(  \left.  P_{n,\theta}\in V_{k,\varepsilon
}\left(  P_{n}\right)  \right\vert P_{n}\right)  \leq-KL_{\mathcal{P}_{k}%
}\left(  V_{k,\varepsilon}\left(  P_{n}\right)  ,P_{\theta}\right)
\]

\end{lem}

Lemmas \ref{Lemma inf} and \ref{Lemma sup} link the Maximum Likelihood
Principle with the Large deviation statements. Define
\begin{equation}
\theta_{ML}:=\arg\max_{\theta}\frac{1}{n}\log P_{\theta}\left(  \left.
P_{n,\theta}\in V_{k,\varepsilon}\left(  P_{n}\right)  \right\vert
P_{n}\right)  \label{ML continuous}%
\end{equation}
and
\[
\theta_{LDP}:=\arg\min_{\theta}-KL_{\mathcal{P}_{k}}\left(  V_{k,\varepsilon
}\left(  P_{n}\right)  ,P_{\theta}\right)
\]
assuming those parameters defined, possibly not in a unique way. Denote%
\[
L_{k,\varepsilon}\left(  \theta\right)  :=\frac{1}{n}\log P_{\theta}\left(
\left.  P_{n,\theta}\in V_{k,\varepsilon}\left(  P_{n}\right)  \right\vert
P_{n}\right)
\]
and%
\[
K_{k,\varepsilon}\left(  \theta\right)  :=-KL_{\mathcal{P}_{k}}\left(
V_{k,\varepsilon}\left(  P_{n}\right)  ,P_{\theta}\right)  .
\]
We then deduce that
\begin{align*}
-\frac{k}{n}\log\left(  n+1\right)   &  \leq L_{k,\varepsilon}\left(
\theta_{ML}\right)  -K_{k,\varepsilon}\left(  \theta_{ML}\right)  \leq0\\
0  &  \leq-L_{k,\varepsilon}\left(  \theta_{LDP}\right)  -K_{k,\varepsilon
}\left(  \theta_{LDP}\right)  \leq\frac{k}{n}\log\left(  n+1\right)
\end{align*}
whence
\begin{equation}
0\leq L_{k,\varepsilon}\left(  \theta_{ML}\right)  -L_{k,\varepsilon}\left(
\theta_{LDP}\right)  \leq\frac{k}{n}\log\left(  n+1\right)
\label{controleML-LDP}%
\end{equation}
from which $\theta_{LDP}$ is a good substitute for $\theta_{ML}$ for fixed $k$
and $\varepsilon$ in the partitioned based model. Note that the bounds in
(\ref{controleML-LDP}) do not depend on the peculiar choice of $\mathcal{P}%
_{k}$ in $\mathfrak{P}_{k}$ .

Fix $k=k_{n}$ such that $\lim_{n\rightarrow\infty}k_{n}=\infty$ together with
$\lim_{n\rightarrow\infty}k_{n}/n=0.$ Define the partition $\mathcal{P}_{k}$
such that $P_{n}(A_{j})=k_{n}/n$ for all $j=1,...,k.$ Hence $A_{j}$ contains
only $k$ sample points. Let $\varepsilon>0$ such that $\max_{1\leq j\leq
k}\left\vert P_{\theta_{T}}(A_{j})-k_{n}/n\right\vert <\varepsilon.$ Then
clearly $P_{\theta_{T}}$ belongs to $V_{k,\varepsilon}\left(  P_{n}\right)  $
and $V_{n,\varepsilon}\left(  P_{n}\right)  $ is included in
$V_{k,2\varepsilon}\left(  P_{\theta_{T}}\right)  .$ Therefore for any
$\theta$ it holds%
\begin{equation}
KL_{\mathcal{P}_{k}}\left(  V_{k,2\varepsilon}\left(  P_{\theta_{T}}\right)
,P_{\theta}\right)  \leq KL_{\mathcal{P}_{k}}\left(  V_{k,\varepsilon}\left(
P_{n}\right)  ,P_{\theta}\right)  \leq KL_{\mathcal{P}_{k}}\left(
P_{\theta_{T}},P_{\theta}\right)  \label{InegDiv}%
\end{equation}
which proves that $\inf_{\theta}$ $KL_{\mathcal{P}_{k}}\left(
V_{k,\varepsilon}\left(  P_{n}\right)  ,P_{\theta}\right)  =0$ with attainment
on $\theta^{\prime}$ such that $P_{\theta^{\prime}}$ and $P_{\theta_{T}}$
coincide on $\mathcal{P}_{k}.$

We now turn to the study of the RHS term in (\ref{InegDiv}). Introducing the
likelihood divergence $\widetilde{\varphi}$ defined in (\ref{divML}) leads%
\[
KL_{\mathcal{P}_{k}}\left(  P_{\theta_{T}},P_{\theta}\right)  =\left(
KL_{m}\right)  _{\mathcal{P}_{k}}\left(  P_{\theta},P_{\theta_{T}}\right)
\]
whence minimizing $KL_{\mathcal{P}_{k}}\left(  P_{\theta_{T}},P_{\theta
}\right)  $ over $\theta$ in $\Theta$ amounts to minimizing the likelihood
divergence $\theta\rightarrow\left(  KL_{m}\right)  _{\mathcal{P}_{k}}\left(
P_{\theta},P_{\theta_{T}}\right)  .$ Set therefore%
\[
\theta_{LDP,\mathcal{P}_{k}}:=\arg\min_{\theta}KL_{\mathcal{P}_{k}}\left(
P_{\theta_{T}},P_{\theta}\right)  =\arg\min_{\theta}\left(  KL_{m}\right)
_{\mathcal{P}_{k}}\left(  P_{\theta},P_{\theta_{T}}\right)  .
\]
Based on the $\sigma-$field generated by $\mathcal{P}_{k}$ on $S$ the dual
form (\ref{dual param}) of the Likelihood divergence pseudodistance $\left(
KL_{m}\right)  _{\mathcal{P}_{k}}\left(  P_{\theta},P_{\theta_{T}}\right)  $
yields
\begin{align}
\arg\min_{\theta}\left(  KL_{m}\right)  _{\mathcal{P}_{k}}\left(  P_{\theta
},P_{\theta_{T}}\right)   &  =\arg\min_{\theta}\sup_{\eta}\sum_{B_{j}%
\in\mathcal{P}_{k}}\widetilde{\varphi}\left(  \frac{P_{\theta}}{P_{\eta}%
}\left(  A_{j}\right)  \right)  P_{\theta}\left(  A_{j}\right) \nonumber\\
&  -\sum_{B_{j}\in\mathcal{P}_{k}}\left(  \widetilde{\varphi}\right)  ^{\ast
}\left(  \frac{P_{\theta}}{P_{\eta}}\left(  A_{j}\right)  \right)
P_{\theta_{T}}\left(  A_{j}\right)  . \label{Vraisemblancedupartition}%
\end{align}
with $\widetilde{\varphi}(x)=-\log x+x-1$ and $\left(  \widetilde{\varphi
}\right)  ^{\ast}(x)=-\log\left(  1-x\right)  .$ With the present choice for
$\widetilde{\varphi}$ the terms in $P_{\eta}$ vanish in the above expression ;
however we complete a full developement, as required in more envolved sampling
schemes. Now an estimate of $\theta_{T}$ is obtained substituting
$P_{\theta_{T}}$ by $P_{n}$ in (\ref{Vraisemblancedupartition}) leading,
denoting $n_{j}$ the number of $X_{i}$'s in $A_{j}$
\[
\widehat{\theta}_{LDP,\mathcal{P}_{k}}:=\arg\min_{\theta}\sup_{\eta}%
\sum_{A_{j}\in\mathcal{P}_{k}}\widetilde{\varphi}\left(  \frac{P_{\theta}%
}{P_{\eta}}\left(  A_{j}\right)  \right)  P_{\theta}\left(  A_{j}\right)
-\sum_{A_{j}\in\mathcal{P}_{k}}\frac{n_{j}}{n}\left(  \widetilde{\varphi
}\right)  ^{\ast}\left(  \frac{P_{\theta}}{P_{\eta}}\left(  A_{j}\right)
\right)  .
\]
Letting $n$ tend to infinity yields (recall that $k=k_{n}$)
\[
\lim_{n\rightarrow\infty}\sup_{\eta}\left\vert
\begin{array}
[c]{c}%
\left[  \sum_{A_{j}\in\mathcal{P}_{k}}\widetilde{\varphi}\left(
\frac{P_{\theta}}{P_{\eta}}\left(  A_{j}\right)  \right)  -\sum_{A_{j}%
\in\mathcal{P}_{k}}\left(  \widetilde{\varphi}\right)  ^{\ast}\left(
\frac{P_{\theta}}{P_{\eta}}\left(  A{j}\right)  \right)  P_{\theta_{T}}\left(
A_{j}\right)  \right] \\
-\left[  \int\widetilde{\varphi}\left(  \frac{p_{\theta}}{p_{\eta}}\left(
x\right)  \right)  p_{\theta}\left(  x\right)  dx-\int\left(  \widetilde
{\varphi}\right)  ^{\ast}\left(  \frac{p_{\theta}}{p_{\eta}}\left(  x\right)
\right)  dP_{n}(x)\right]
\end{array}
\right\vert =0
\]
w.p. $1$ which in turn implies
\[
\lim_{n\rightarrow\infty}\widehat{\theta}_{LDP,\mathcal{P}_{k}}-\widehat
{\theta}_{ML}=0
\]
where $\widehat{\theta}_{ML}$ is readily seen to be the usual ML estimator of
$\theta$ defined through
\[
\widehat{\theta}_{ML}:=\arg\sup_{\theta}%
%TCIMACRO{\dprod \limits_{i=1}^{n}}%
%BeginExpansion
{\displaystyle\prod\limits_{i=1}^{n}}
%EndExpansion
p_{\theta}\left(  X_{i}\right)  .
\]

\section{Weighted sampling}

This section extends the previous arguments for weighted sampling schemes. We
will show that the Maximum Likelihood paradigm as defined above can be
extended for these schemes, leading to operational procedures involving the
minimization of specific divergence pseudodistances defined in strong relation
with the distribution of the weights.

The sampling scheme which we consider is commonly used in connection with the
bootstrap and is refered to as the \textit{weighted} or \textit{generalized
bootstra}p, sometimes called \textit{wild bootstrap}, first introduced by
Newton and Mason \cite{Mason}. The main simplification which we consider in
the present setting lies in the fact that we assume that the weights $W_{i}$
are i.i.d. while being exchangeable random variables in the generalized
bootstrap setting.

Let $x_{1},...,x_{n}$ be $n$ independent realizations of $n$ i.i.d. r.v's
$X_{1},...,X_{n}$ with common distribution $P_{\theta_{T}.}$ It will be
assumed that
\begin{equation}
\text{For all }\theta\text{ in }\Theta\text{, }E_{\theta}X\text{ and
}E_{\theta}X^{2}\text{ are finite.} \label{EXet EX^2 finis}%
\end{equation}
This entails that both
\[
\frac{1}{n}\sum_{i=1}^{n}x_{i}\text{ and }\frac{1}{n}\sum_{i=1}^{n}x_{i}^{2}%
\]
converge $P_{\theta_{T}}-$a.e. to $E_{\theta_{T}}X$ and $E_{\theta_{T}}X^{2}$
respectively; also the same holds with $\theta_{T}$ substituted by any
$\theta$ in $\Theta$ when $x_{1},...,x_{n}$ is sampled under $P_{\theta}.$
This assumption is necessary when studying the properties of the estimates of
$\theta_{T}$ and of $\phi\left(  \theta_{T},\theta\right)  $ under some
alternative $\theta.$

Consider a collection $W_{1},...,W_{n}$ of independent copies of $W$, whose
distribution satisfies the conditions stated in Section 1. The weighted
empirical measure $P_{n}^{W}$ is defined through%

\[
P_{n}^{W}:=\frac{1}{n}\sum_{i=1}^{n}W_{i}\delta_{x_{i}}.
\]
This empirical measure need not be a probability measure, since its mass may
not equal $1.$ Also it might not be positive, since the weights may take
negative values.\ The measure $P_{n}^{W}$ converges almost surely to
$P_{\theta_{T}}$ when the weights $W_{i}$'s satisfy the hypotheses stated in
Section 1. Indeed general results pertaining to this sampling procedure state
that under regularity, functionals of the measure $P_{n}^{W}$ are
asymptotically distributed as are the same functionals of $P_{n}$ when the
$X_{i}$'s are i.i.d. Therefore the weighted sampling procedure mimicks the
i.i.d. sampling fluctuation in a two steps procedure: choose $n$ values of
$x_{i}$ such that they asymptotically fit to $P_{\theta_{T}}$, which means%
\[
\lim_{n\rightarrow\infty}\frac{1}{n}\sum_{i=1}^{n}\delta_{x_{i}}=P_{\theta
_{T}}%
\]
deterministically and then play the $W_{i}$'s on each of the $x_{i}$'s. Then
get $P_{n}^{W}$, a proxy to the random empirical measure $P_{n}$ .

For any $\theta$ in $\Theta$ consider a similar sampling procedure under the
weights $W_{i}^{\prime}$ 's \ which are i.i.d. copies of the $W_{i}$'s. Let
therefore $x_{1,\theta},...,x_{n,\theta}$ denote $n$ i.i.d. realizations of
$X_{1,\theta},...,X_{n,\theta}$ with distribution $P_{\theta}$ yielding the
empirical measure%

\[
P_{n,\theta}^{W^{\prime}}:=\frac{1}{n}\sum_{i=1}^{n}W_{i}^{\prime}%
\delta_{x_{i,\theta}}%
\]
the corresponding empirical measure. Note that except for the choice of the
generating measure $P_{\theta}$ , $P_{n,\theta}^{W^{\prime}}$ is obtained in
the same way as $P_{n}^{W}.$ The ML principle turns out to select the value of
$\theta$ making $P_{n,\theta}^{W^{\prime}}$ as close as possible from
$P_{n}^{W},$ conditionally upon $P_{n}^{W}.$
%Denote $P_{\theta,W^{\prime}}$
%the product measure $P_{\theta}\times P_{W^{\prime}}$ where $P_{W^{\prime}}$
%is the distribution of $W^{\prime}.$

The resulting estimates are optimal in many respects, as is the classical ML
estimator for regular models in the i.i.d. sampling scheme.\ The proposal
which is presented here also allows to obtain optimal estimators for some non
regular models. This approach is in line with \cite{BK2009} who developped a
whole range of first order optimal estimation procedures in the case of the
i.i.d. sampling, based on divergence minimization.

Using the notations of section $\ref{measureSpace}$, we endow $\mathcal{M}(S)$
with $\tau_{0}$-topology rather than the weak topology, and define accordingly
the $\sigma$-field $\mathcal{B}(\mathcal{M})$ on $\mathcal{M}(S)$. Denote by
$\mathcal{M}_{1}(S)$ the space of probability measure on $S,$ endowed with the
$\tau_{0}-$topology.

\subsection{A Sanov conditional theorem for the weighted empirical measure}

\label{sectSav}

The procedure which we are going to develop can be stated as follows.

Similarly as in the simple i.i.d. setting select some (small) neighborhood
$V_{\epsilon}\left(  P_{n}^{W}\right)  $of $P_{n}^{W}$ and define the MLE of
$\theta_{T}$ as the value of $\theta$ which optimizes the probability that the
simulated empirical measure $P_{n,\theta}^{W^{\prime}}$ belongs to
$V_{\epsilon}\left(  P_{n}^{W}\right)  .$ This requires a conditional Sanov
type result, substituting Lemmas $\ref{Lemma inf}$ and $\ref{Lemma sup}$. This
result is produced in Theorem $\ref{ConditionedLDP}$ in Section $\ref{sectSav}%
$. In the same vein as in Lemmas $\ref{Lemma inf}$ and $\ref{Lemma sup}$,
maximizing in $\theta$ this probability amounts to minimizing a LDP rate
between $P_{\theta}$ and $V_{\epsilon}\left(  P_{\theta_{T}}\right)  .$ The
rate is in strong relation with the distribution of the $W_{i}$'s. Call it
$\phi^{W}\left(  V_{\epsilon}\left(  P_{\theta_{T}}\right)  ,P_{\theta
}\right)  :=\inf\left\{  \phi^{W}\left(  Q,P_{\theta}\right)  ,Q\in
V_{\epsilon}\left(  P_{\theta_{T}}\right)  \right\}  .$

Since $\epsilon$ is small, this rate is of order $\phi^{W}\left(
P_{\theta_{T}},P_{\theta}\right)  ;$ this is Corollary $\ref{approximation}$
in Section $\ref{sectSav}$. Turn to the original data and estimate $\phi
^{W}\left(  P_{\theta_{T}},P_{\theta}\right)  $ by some plug in method to be
stated in Section $\ref{SecDiv}$. Define the ML estimator of $\theta_{T}$
through the minimization of the proxy of $\phi^{W}\left(  P_{\theta_{T}%
},P_{\theta}\right)  .$ We will prove that minimum divergence estimators play
a key role in this setting.

In order to state our conditional Sanov theorem we put forwards the following
lemma, which is in the vein of Theorem 2.2 of Najim $\cite{Najim}$ which
states the Sanov large deviation theorem, where the weights are i.i.d random
variables. Trashorras and Wintenberger $\cite{TraWinten}$ have investigated
the large deviations properties of weighted (bootstrapped) empirical measure
with exchangeable weights under appropriate assumptions of the weights. Both
papers equip $\mathcal{M}(S)$ with the weak topology.

The lemma's proof is defered to Section \ref{proof}.\bigskip

\begin{lem}
\label{WeightedLDP} Assume that $P_{\theta}(U)>0$ for any non-empty open set
$U\in S$, and that $\lim_{n\rightarrow\infty}P_{n}=\lim_{n\rightarrow\infty
}\frac{1}{n}\sum_{i=1}^{n}\delta_{x_{i}}=P_{\theta}\in\mathcal{M}_{1}(S)$,
where the convergence holds under $\tau_{0}.$ Then $P_{n,\theta}^{W}$
satisfies the LDP in $(\mathcal{M}(S),\mathcal{B}(\mathcal{M}))$ equipped with
the $\tau_{0}$-topology with the good convex rate function:
\begin{align*}
\phi^{W}(\zeta,P_{\theta})  &  =\sup_{f\in B(\mathbb{R}^{d})}\Big\{\int
_{\mathbb{R}^{d}}f(x)\zeta(dx)-\int_{\mathbb{R}^{d}}M(f(x))P_{\theta
}(dx)\Big\}\\
&  =%
\begin{cases}
\int_{\mathbb{R}^{d}}M^{\ast}(\frac{d\zeta}{dP_{\theta}})dP_{\theta}%
,\qquad\qquad & \text{if }\zeta\text{ is a.c. w.r.t. }P_{\theta}\\
\infty, & \text{otherwise}%
\end{cases}
\end{align*}
where $M^{\ast}(x)=\sup_{t}tx-M(t)$ for all real $x$ and $M(t)$ is the moment
generating function of $W.$
\end{lem}

\bigskip

Let $\mathcal{P}_{k}=\left(  A_{1},...,A_{k}\right)  $ denote an arbitrary
partition of $S$ with $A_{i}$ in $B(S)$ for all $i=1,...,k\mathbb{\ }$, and
define the pseudometric $d_{\mathcal{P}_{k}}$ on $\mathcal{M}(S)$ by
\[
d_{\mathcal{P}_{k}}(Q,R)=\max_{1\leq j\leq k}|Q(B_{j})-R(B_{j})|,\qquad
Q,R\in\mathcal{M}(S).
\]
For any positive $\epsilon,$ let
\[
V_{\epsilon}(P_{n}^{W})=\{Q\in\mathcal{M}(S):d_{\mathcal{P}_{k}}(Q,P_{n}%
^{W})<\epsilon\}
\]
denote an open neighborhood of the weighted empirical measure $P_{n}^{W}$ in
the $\tau_{0}$ -topology$.$ Then we have the following conditional LDP theorem.

\begin{theo}
\label{ConditionedLDP} With the above notation and assuming that
$P_{\theta_{T}}$ is absolutely continuous with respect to $P_{\theta}$, for
any positive $\epsilon$, the following conditional LDP result holds%

\begin{align*}
\lim_{n\rightarrow\infty}\frac{1}{n}\log P_{\theta}\Big(P_{n,\theta
}^{W^{\prime}}\in V_{\epsilon}(P_{n}^{W})|P_{n}\Big)= -\phi^{W}(V_{\epsilon
}(P_{\theta_{T}}),P_{\theta}).
\end{align*}

\end{theo}

\begin{proof} In the following proof, $\mathcal{P}_{k}$ is an arbitrary partition on
$S.$
\begin{align*}
P_{\theta}\Big(P_{n,\theta}^{W^{\prime}}  &  \in V_{\epsilon}(P_{n}^{W}%
)|P_{n}\Big)=P_{\theta}\Big(d_{\mathcal{P}_{k}}(P_{n,\theta}^{W^{\prime}},
P_{n}^{W})<\epsilon|P_{n}\Big)\\
&  \geq P_{\theta}\Big(d_{\mathcal{P}_{k}}(P_{n,\theta}^{W^{\prime}}%
,P_{\theta_{T}})+d_{\mathcal{P}_{k}}(P_{\theta_{T}},P_{n}^{W})<\epsilon
|P_{n}\Big)\\
&  =P_{\theta}\Big(d_{\mathcal{P}_{k}}(P_{n,\theta}^{W^{\prime}},P_{\theta
_{T}})<\epsilon-d_{\mathcal{P}_{k}}(P_{\theta_{T}},P_{n}^{W})|P_{n}\Big).
\end{align*}
Since $d_{\mathcal{P}_{k}}(P_{\theta_{T}},P_{n}^{W})\rightarrow0$ when $n$
$\rightarrow\infty$, for any positive $\delta$ and sufficiently large $n$ we
have:
\begin{align*}
P_{\theta}\Big(P_{n,\theta}^{W^{\prime}}\in V_{\epsilon}(P_{n}^{W}%
)|P_{n}\Big)  &  \geq P_{\theta}\Big(d_{\mathcal{P}_{k}}(P_{n,\theta
}^{W^{\prime}},P_{\theta_{T}})<\epsilon-\delta\Big) =P_{\theta}%
\Big(P_{n,\theta}^{W^{\prime}}\in V_{\epsilon-\delta}(P_{\theta_{T}})\Big).
\end{align*}
By Lemma $\ref{WeightedLDP}$, we obtain the conditioned LDP lower bound
\[
\liminf_{n\rightarrow\infty}\frac{1}{n}\log P_{\theta}\Big(P_{n,\theta
}^{W^{\prime}}\in V_{\epsilon}(P_{n}^{W})|P_{n}\Big)\geq-\phi^{{W}%
}(V_{\epsilon-\delta}(P_{\theta_{T}}),P_{\theta}),
\]
In a similar way, we obtain the large deviation upper bound
\begin{align*}
&  P_{\theta}\Big(P_{n,\theta}^{W^{\prime}}\in V_{\epsilon}(P_{n}^{W}%
)|P_{n}\Big)=P_{\theta}\Big(d_{\mathcal{P}_{k}}(P_{n,\theta}^{W^{\prime}%
},P_{n}^{W})<\epsilon|P_{n}\Big)\\
&  \leq P_{\theta}\Big(d_{\mathcal{P}_{k}}(P_{n,\theta}^{W^{\prime}}%
,P_{\theta_{T}})-d_{\mathcal{P}_{k}}(P_{\theta_{T}},P_{n}^{W})<\epsilon
|P_{n}\Big)\\
&  \leq P_{\theta}\Big(d_{\mathcal{P}_{k}}(P_{n,\theta}^{W^{\prime}}%
,P_{\theta_{T}})<\epsilon+\delta^{\prime}\Big) =P_{\theta}\Big(P_{n,\theta
}^{W^{\prime}}\in V_{\epsilon+\delta^{\prime}}(P_{\theta_{T}})\Big),
\end{align*}
for some positive $\delta^{\prime}.$ We thus obtain
\begin{align*}
\limsup_{n\rightarrow\infty}\frac{1}{n}\log P_{\theta}\Big(P_{n,\theta
}^{W^{\prime}}  &  \in V_{\epsilon}(P_{n}^{W})|P_{n}\Big)\leq-\phi^{{W}%
}(V_{\epsilon+\delta^{\prime}}(P_{\theta_{T}}),P_{\theta}).
\end{align*}
Let $\delta^{\prime\prime}=max(\delta,\delta^{\prime})$, we have
\begin{align*}
-\phi^{{W}}(V_{\epsilon-\delta^{\prime\prime}}(P_{\theta_{T}}),P_{\theta})  &
\leq\liminf_{n\rightarrow\infty}\frac{1}{n}\log P_{\theta}\Big(P_{n,\theta
}^{W^{\prime}}\in V_{\epsilon}(P_{n}^{W})|P_{n}\Big)\\
&  \leq\limsup_{n\rightarrow\infty}\frac{1}{n}\log P_{\theta}\Big(P_{n,\theta
}^{W^{\prime}}\in V_{\epsilon}(P_{n}^{W})|P_{n}\Big)\leq-\phi^{{W}%
}(V_{\epsilon+\delta^{\prime\prime}}(P_{\theta_{T}}),P_{\theta}).
\end{align*}
Denote $cl_{\tau_{0}}(V_{\epsilon}(P_{\theta_{T}}))$ the closure of the open
set $V_{\epsilon}(P_{\theta_{T}})$ in the $\tau_{0}$-topology, and note
$\delta^{\prime\prime}$ is arbitrarily small, then it holds
\begin{align*}
-\phi^{W}(V_{\epsilon}(P_{\theta_{T}}),P_{\theta})  &  \leq\liminf
_{n\rightarrow\infty}\frac{1}{n}\log P_{\theta}\Big(P_{n,\theta}^{W^{\prime}%
}\in V_{\epsilon}(P_{n}^{W})|P_{n}\Big)\\
&  \leq\limsup_{n\rightarrow\infty}\frac{1}{n}\log P_{\theta}\Big(P_{n,\theta
}^{W^{\prime}}\in V_{\epsilon}(P_{n}^{W})|P_{n}\Big)\leq-\phi^{{W}%
}\big(cl_{\tau_{0}}(V_{\epsilon}(P_{\theta_{T}})),P_{\theta}\big).
\end{align*}
It remains to show that
\begin{align}\label{phivoig}
\phi^{W}(V_{\epsilon}(P_{\theta_{T}}),P_{\theta})=\phi^{{W}}\big(cl_{\tau_{0}%
}(V_{\epsilon}(P_{\theta_{T}})),P_{\theta}\big).
\end{align}
Since $P_{\theta_{T}}$ is absolutely continuous with respect to $P_{\theta}$,
by Lemma \ref{WeightedLDP} we have
\begin{align}
\phi^{{W}}\big(cl_{\tau_{0}}(V_{\epsilon}(P_{\theta_{T}})),P_{\theta}%
\big)\leq\phi^{W}(V_{\epsilon}(P_{\theta_{T}}),P_{\theta})\leq\phi
^{W}(P_{\theta_{T}},P_{\theta})<\infty. \label{phifinite}
\end{align}
Given some small positive constant $\omega$, then there exists $\mu\in
cl_{\tau_{0}}(V_{\epsilon}(P_{\theta_{T}}))$ satisfying
\[
\phi^{W}(\mu,P_{\theta})<\phi^{{W}}\big(cl_{\tau_{0}}(V_{\epsilon}%
(P_{\theta_{T}})),P_{\theta}\big)+\omega.
\]
Set $v\in V_{\epsilon}(P_{\theta_{T}})$, and define $z(\alpha)=\alpha
\mu+(1-\alpha)v$, where $0<\alpha<1$. Obviously, we have $z(\alpha)\in
V_{\epsilon}(P_{\theta_{T}})$. By Lemma $\ref{WeightedLDP}$, the map
$\zeta\rightarrow\phi(\zeta,P_{\theta})$ is convex, hence we get
\begin{align}
\phi^{W}(V_{\epsilon}(P_{\theta_{T}}),P_{\theta})  &  \leq\lim_{\alpha
\rightarrow1}\phi^{W}(z(\alpha),P_{\theta})\leq\lim_{\alpha\rightarrow
1}\Big(\alpha\phi^{W}(\mu,P_{\theta})+(1-\alpha)\phi^{W}(v,P_{\theta
})\Big)\nonumber\label{taufermeture}\\
&  =\phi^{W}(\mu,P_{\theta})<\phi^{{W}}\big(cl_{\tau_{0}}(V_{\epsilon
}(P_{\theta_{T}})),P_{\theta}\big)+\omega,
\end{align}
where the equality holds since $\phi^{W}(v,P_{\theta})$ is finite by
$(\ref{phifinite})$. Combine $(\ref{phifinite})$ with $(\ref{taufermeture})$
to get $(\ref{phivoig})$. This proves the conditional large deviation result.
\end{proof}
\bigskip

Using the above theorem, we obtain the following corollary.

\begin{cor}
$\label{approximation}$ Under the assumptions of Theorem $\ref{ConditionedLDP}%
$, it holds
\begin{align*}
\lim_{\epsilon\rightarrow0} \phi^{W}(V_{\epsilon}(P_{\theta_{T}}),P_{\theta
})=\phi(P_{\theta_{T}},P_{\theta}).
\end{align*}

\end{cor}

\begin{proof} By Lemma \ref{WeightedLDP}, the rate function $\phi^{{W}}(\mu
,P_{\theta})$ is a good rate function, hence it is lower semi-continuous; this
implies
\begin{equation}\label{fjl}
\lim_{\epsilon\rightarrow0}\phi^{W}(V_{\epsilon}(P_{\theta_{T}}),P_{\theta
})\geq\phi(P_{\theta_{T}},P_{\theta}).
\end{equation}
For any $\epsilon>0$, we have $\phi^{{W}}(P_{\theta_{T}},P_{\theta})\geq
\phi^{{W}}(V_{\epsilon}(P_{\theta_{T}}),P_{\theta});$ this together with
$(\ref{fjl})$ completes the proof.
\end{proof}
\bigskip

\subsection{Divergences associated to the weighted sampling scheme}

\label{SecDiv}

For any $Q$ in $V_{\epsilon}(P_{\theta_{T}})$ rewrite the good rate function
using the divergence notation
\begin{equation}
\phi^{W}(Q,P_{\theta})=\int M^{\ast}\left(  \frac{dQ}{dP_{\theta}}\right)
dP_{\theta} =\int\varphi^{W}\left(  \frac{dQ}{dP_{\theta}}\right)  dP_{\theta}
\label{divM*}%
\end{equation}
from which $\phi^{W}\left(  Q,P_{\theta}\right)  $ is the divergence
associated with the divergence function $\varphi^{W}:=M^{\ast}.$

Commuting $P_{\theta_{T}}$ and $P_{\theta}$ in (\ref{divM*}) and introducing
the conjugate divergence function $\widetilde{\varphi^{W}}$ yields
\begin{align}
\label{conjDiv}\phi^{W}(Q,P_{\theta})=\int\varphi^{W}\left(  \frac
{dQ}{dP_{\theta}}\right)  dP_{\theta}=\int\widetilde{\varphi^{W}}\left(
\frac{dP_{\theta}}{dQ}\right)  dQ=\widetilde{\phi^{W}}(P_{\theta},Q).
\end{align}

By Theorem $\ref{ConditionedLDP}$, maximizing $P_{\theta}(P_{n,\theta
}^{W^{\prime}}\in V_{\epsilon}(P_{n}^{W})|P_{n})$ amounts to minimize
${\phi^{W}}(V_{\epsilon}(P_{\theta_{T}}),P_{\theta}).$ A final approximation
now yields the form of the criterion to be estimated in order to define the
MLE in the present setting. As $\epsilon\rightarrow0$ the asymptotic order of
$\phi^{W}(V_{\epsilon}(P_{\theta_{T}}),P_{\theta})$ is equal to $\widetilde
{\phi^{W}}(P_{\theta},P_{\theta_{T}})$ by Corollary $\ref{approximation}$ and
$(\ref{conjDiv})$, which is a proxy of $\phi^{W}(P_{\theta_{T}},P_{\theta})$
and therefore the theoretical criterion to be optimized in $\theta.$

We now state the dual form of the theoretical criterion $\widetilde{\phi^{W}%
}(P_{\theta},P_{\theta_{T}})$ using the dual form (\ref{dual param}) and
(\ref{dual h}). It holds%
\begin{equation}
\widetilde{\phi^{W}}(P_{\theta},P_{\theta_{T}})=\sup_{\alpha\in\mathcal{U}%
}\int\widetilde{h}(\theta,\alpha,x)dP_{\theta_{T}}(x) \label{criterepondere}%
\end{equation}
with
\[
\widetilde{h}(\theta,\alpha,x)=\int\left(  \widetilde{\varphi^{W}}\right)
^{\prime}\left(  \frac{dP_{\theta}}{dP_{\alpha}}\right)  dP_{\theta}-\left(
\widetilde{\varphi^{W}}\right)  ^{\#}\left(  \frac{dP_{\theta}}{dP_{\alpha}%
}\left(  x\right)  \right)
\]

We now turn to the definition of the MLE in this context, estimating the
criterion and deriving the estimate.

\subsection{MLE under weighted sampling}

Using the dual representation of divergences, the natural estimator of
$\phi(P_{\theta},P_{\theta_{T}})$ is%
\begin{equation}
\widetilde{\phi_{n}}(P_{\theta},P_{\theta_{T}}):=\sup_{\alpha\in\mathcal{U}%
}\left\{  \int\widetilde{h}(\theta,\alpha,x)~dP_{n}^{W}(x)\right\}  .
\label{estim div duale}%
\end{equation}

From now on, we will use $\phi(\theta,{\theta_{T}})$ to denote $\phi
(P_{\theta},P_{\theta_{T}});$ whence the resulting estimator of $\phi
(\theta_{T},\theta_{T})$ $\ $is
\[
\widetilde{\phi_{n}}(\theta_{T},\theta_{T}):=\inf_{\theta\in\Theta}%
\widetilde{\phi_{n}}(\theta,\theta_{T})=\inf_{\theta\in\Theta}\sup_{\alpha
\in\mathcal{U}}\left\{  \int\widetilde{h}(\theta,\alpha,x)~dP_{n}%
^{W}(x)\right\}
\]
and the resulting MLE\ of $\theta_{T}$ is obtained as the minimum dual
$\widetilde{\phi^{W}}$ estimator

$\bigskip$%
\begin{equation}
\widehat{\theta}_{ML,W}:=\arg\inf_{\theta\in\Theta}\sup_{\alpha\in\mathcal{U}%
}\left\{  \int\widetilde{h}(\theta,\alpha,x)~dP_{n}^{W}(x)\right\}  .
\label{def EMphiD estimates}%
\end{equation}

Formula (\ref{def EMphiD estimates}) indeed defines a Maximum Likelihood
estimator, in the vein of (\ref{MLfinite}) and (\ref{ML continuous}). This
estimator requires no grouping nor smoothing.

\section{Bahadur slope of minimum divergence tests for weighted data}

Consider the test of some null hypothesis H0: $\theta_{T}=\theta$ versus a
simple hypothesis H1 $\theta_{T}=\theta^{\prime}.$

We consider two competitive statistics for this problem. The first one is
based on the estimate of $\widetilde{\phi}^{W}\left(  P_{\alpha},P_{\beta
}\right)  $ defined for all $\left(  \alpha,\beta\right)  $ in $\Theta
\times\Theta$ through
\[
T_{n}\left(  \alpha\right)  :=\sup_{\eta\in\Theta}\int\widetilde{\varphi}%
^{W}\left(  \frac{p_{\alpha}}{p_{\eta}}\right)  p_{\eta}d\mu-\int\left(
\widetilde{\varphi}^{W}\right)  ^{\ast}\left(  \frac{p_{\alpha}}{p_{\beta}%
}\right)  dP_{n}^{W}
\]
where the i.i.d. sample $X_{1},...,X_{n}$ has distribution $P_{\beta}.$ The
test statistics $T_{n}\left(  \theta\right)  $ converges to $0$ under H0.

A competitive statistics $\widehat{\psi}\left(  \theta\right)  $ writes
\[
\widehat{\psi}\left(  \theta\right)  :=\psi\left(  \theta,P_{n}^{W}\right)
\]
where $Q\rightarrow$ $\psi\left(  \theta,Q\right)  $ is assumed to satisfy
$\psi\left(  \theta,P_{\theta}\right)  =0$ $,$ and is $\tau-$continuous with
respect to $Q$, which implies that under H0 the following Large Deviation
Principle holds
\begin{align}
\lim_{n\rightarrow\infty}\frac{1}{n}\log P_{\theta}\left(  \widehat{\psi
}\left(  \theta\right)  \geq t\right)   &  =-I(t)\label{LDP sous H0}\\
&  =-\inf\left\{  \phi^{W}\left(  P_{\theta},Q\right)  ,\psi\left(
\theta,Q\right)  \geq t\right\} \nonumber
\end{align}
for any positive $t.$ Also we assume that under H1,$\widehat{\psi}\left(
\theta\right)  $ converges to $\psi\left(  \theta,P_{\theta^{\prime}}\right)
$
\begin{equation}
\lim_{n\rightarrow\infty}\widehat{\psi}\left(  \theta\right)  =_{\theta
^{\prime}}\psi\left(  \theta,P_{\theta^{\prime}}\right)  \label{cv sous teta'}%
\end{equation}
where (\ref{cv sous teta'}) stands in probability under $\theta^{\prime}$.

We now state the Bahadur slope of the test $\widehat{\phi^{W}}\left(
\theta,\theta\right)  .$

Under H0%
\begin{align*}
\lim_{n\rightarrow\infty}\frac{2}{n}\log P_{\theta}\left(  T_{n}\left(
\theta\right)  \geq t\right)   &  =-2\inf\left\{  \phi^{W}\left(  P_{\theta
},Q\right)  ,\widetilde{\phi}^{W}\left(  Q,P_{\theta}\right)  \geq t\right\}
\\
&  =-2\inf\left\{  \phi^{W}\left(  P_{\theta},Q\right)  ,\phi^{W}\left(
P_{\theta},Q\right)  \geq t\right\} \\
&  =-2t
\end{align*}
while, under H1%
\[
\lim_{n\rightarrow\infty}T_{n}\left(  \theta\right)  =\phi^{W}\left(
P_{\theta},P_{\theta^{\prime}}\right)  \text{ in probability}%
\]
since $P_{n}^{W}$ converges weakly to $P_{\theta^{\prime}}.$

It follows that the Bahadur slope of the minimum divergence test
$\widehat{\phi^{W}}\left(  \theta,\theta\right)  $ is
\[
\ e_{T_{n}\left(  \theta\right)  }=-2\phi^{W}\left(  P_{\theta},P_{\theta
^{\prime}}\right)  .
\]

Let us evaluate the Bahadur slope of the test $\widehat{\psi}\left(
\theta\right)  .$

Following (\ref{LDP sous H0}) and (\ref{cv sous teta'}) it holds%
\[
e_{\widehat{\psi}(\theta)}=-2\inf\left\{  \phi^{W}\left(  P_{\theta},Q\right)
,\psi\left(  \theta,Q\right)  \geq\psi\left(  \theta,P_{\theta^{\prime}%
}\right)  \right\}  .
\]
Since $\inf\left\{  \phi^{W}\left(  P_{\theta},Q\right)  ,\psi\left(
\theta,Q\right)  \geq\psi\left(  \theta,P_{\theta^{\prime}}\right)  \right\}
\leq\phi^{W}\left(  P_{\theta},P_{\theta^{\prime}}\right)  $ it follows that
$e_{\widehat{\psi}\left(  \theta\right)  }\leq e_{T_{n}\left(  \theta\right)
}.$

We have proved

\begin{proposition}
Under the weighted sampling the test statistics $\widehat{\psi}\left(
\theta\right)  $ is Bahadur efficient among all tests which are empirical
versions of $\tau_{0}-$ continuous functionals.
\end{proposition}

\section{Weighted sampling in exponential families}

In this short section we show that MLE's associated with weighted sampling are
specific with respect to the weighting; this is in contrast with the
unweighted sampling (i.i.d. simple sampling), under which all minimum
divergence estimators coincide with the standard MLE; see \cite{BrArxiv}.

Let
\begin{equation}
p_{\theta}(x)=\exp\left[  \theta t(x)-C(\theta)\right]  d\mu(x)
\label{fam exp}%
\end{equation}
be an exponential family with natural parameter $\theta$ in an open set
$\Theta$ in $\mathbb{R}^{d},$ and where $\mu$ denotes a common dominating
measure for the model. We assume that this family is full i.e. that the
Hessian matrix $\left(  \partial^{2}/\partial\theta^{2}\right)  C(\theta)$ is
definite positive.\ Recall that under the standard i.i.d. $X_{1},...,X_{n}$
sampling the MLE $\theta_{ML}$ of $\theta$ satisfies
\[
\nabla C(\theta)_{\theta_{ML}}=\frac{1}{n}\sum_{i=1}^{n}t\left(  X_{i}\right)
.
\]
Under the weighted sampling $W_{1}$ $,...,W_{n}$ corresponding to the
divergence $\ $function $\varphi^{W}$, conditionally on the observed data
$x_{1},...,x_{n}$ the MLE writes%

\[
\theta_{ML,W}:=\arg\inf_{\theta\in\Theta}\sup_{\alpha\in\mathcal{U}}%
\int\left(  \widetilde{\varphi^{W}}\right)  ^{\prime}\left(  \frac{p_{\theta}%
}{p_{\alpha}}\right)  p_{\theta}d\mu-\int\left(  \widetilde{\varphi^{W}%
}\right)  ^{\#}\left(  \frac{p_{\theta}}{p_{\alpha}}\right)  dP_{n}^{W}.
\]
We prove that $\theta_{ML,W}$ satisfies
\[
\nabla C(\theta)_{\theta_{ML,W}}=\frac{1}{n}\sum_{i=1}^{n}W_{i}t\left(
x_{i}\right)  .
\]

Denote
\[
M_{n}\left(  \theta,\alpha\right)  :=\int\left(  \widetilde{\varphi^{W}%
}\right)  ^{\prime}\left(  \frac{p_{\theta}}{p_{\alpha}}\right)  p_{\theta
}d\mu-\int\left(  \widetilde{\varphi^{W}}\right)  ^{\#}\left(  \frac
{p_{\theta}}{p_{\alpha}}\right)  dP_{n}^{W}.
\]
Clearly, subsituting using (\ref{fam exp}) it holds for all $\theta$
\begin{equation}
\inf_{\theta\in\Theta}\sup_{\alpha\in\mathcal{U}}M_{n}\left(  \theta
,\alpha\right)  \geq M_{n}\left(  \theta,\theta\right)  =0. \label{ineg sup}%
\end{equation}
We prove that $M_{n}\left(  \theta_{ML,W},\alpha\right)  $ is maximal for
$\alpha=\theta_{ML,W}$ which closes the proof.

Let $X_{1},...,X_{n}$ be $n$ \ i.i.d. random variables with common
distribution $P_{\theta_{T}\text{ }}$ with $\theta_{T}$ in $\Theta.$
Introduce
\[
M_{n}\left(  \theta,\alpha\right)  :=\int\varphi^{\prime}\left(
\frac{dP_{\theta}}{dP_{\alpha}}\right)  dP_{\theta}-\frac{1}{n}\sum_{i=1}%
^{n}\varphi^{\#}\left(  \frac{dP_{\theta}}{dP_{\alpha}}\left(  X_{i}\right)
\right)
\]

\bigskip We prove that
\begin{equation}
\alpha=\theta_{ML,W}\text{ is the unique maximizer of }M_{n}\left(
\theta_{ML,W},\alpha\right)  \label{alfaMaximise}%
\end{equation}
which yields
\begin{equation}
\inf_{\theta}\sup_{\alpha}M_{n}\left(  \theta,\alpha\right)  \leq\sup_{\alpha
}M_{n}\left(  \theta_{ML,W},\alpha\right)  =M_{n}\left(  \theta_{ML,W}%
,\theta_{ML,W}\right)  =0 \label{ineg inf}%
\end{equation}

which together with (\ref{ineg sup})\bigskip\ completes the proof.

Define
\begin{align*}
M_{n,1}\left(  \theta,\alpha\right)   &  :=\int\varphi^{\prime}\left(  \exp
A(\theta,\alpha,x)\right)  \exp B\left(  \theta,x\right)  d\lambda(x)\\
M_{n,2}\left(  \theta,\alpha\right)   &  :=\frac{1}{n}\sum_{i=1}^{n}W_{i}%
\exp\left(  A\left(  \theta,\alpha,x_{i}\right)  \right)  \varphi^{\prime
}\left(  \exp A(\theta,\alpha,x_{i})\right) \\
M_{n,3}\left(  \theta,\alpha\right)   &  :=\frac{1}{n}\sum_{i=1}^{n}%
W_{i}\varphi\left(  \exp A(\theta,\alpha,x_{i})\right)
\end{align*}
with
\begin{align*}
A(\theta,\alpha,x)  &  :=T(x)^{\prime}\left(  \theta-\alpha\right)
+C(\alpha)-C(\theta)\\
B(\theta,x)  &  :=T(x)^{\prime}\theta-C(\theta).
\end{align*}

It holds
\[
M_{n}\left(  \theta,\alpha\right)  =M_{n,1}\left(  \theta,\alpha\right)
-M_{n,2}\left(  \theta,\alpha\right)  +M_{n,3}\left(  \theta,\alpha\right)
\]
with
\[
\frac{\partial}{\partial\alpha}M_{n,1}\left(  \theta,\alpha\right)
_{\alpha=\theta}=-\varphi^{(2)}\left(  1\right)  \left[  \nabla C\left(
\theta\right)  -\nabla C\left(  \alpha\right)  _{\alpha=\theta}\right]  =0
\]
for all $\theta,$%
\[
\frac{\partial}{\partial\alpha}M_{n,2}\left(  \theta,\alpha\right)
_{\alpha=\theta_{ML,W}}=\varphi^{(2)}\left(  1\right)  \frac{1}{n}\sum
_{i=1}^{n}W_{i}\left[  -T(x_{i})+\nabla C\left(  \alpha\right)  _{\alpha
=\theta_{ML,W}}\right]  =0
\]
and
\[
\frac{\partial}{\partial\alpha}M_{n,3}\left(  \theta_{ML,W},\alpha\right)
=\frac{1}{n}\sum_{i=1}^{n}W_{i}\left[  -T(x_{i})+\nabla C\left(
\alpha\right)  _{\alpha=\theta_{ML,W}}\right]  =0
\]
where the two last displays hold iff $\alpha=\theta_{ML}.$ Now
\begin{align*}
\frac{\partial^{2}}{\partial\alpha^{2}}M_{n,1}\left(  \theta,\alpha\right)
_{\alpha=\theta_{ML,W}}  &  =\left(  \varphi^{(3)}(1)+2\varphi^{(2)}%
(1)\right)  \left(  \partial^{2}/\partial\theta^{2}\right)  C(\theta_{ML,W})\\
\frac{\partial^{2}}{\partial\alpha^{2}}M_{n,2}\left(  \theta,\alpha\right)
_{\alpha=\theta_{ML,W}}  &  =\left(  \varphi^{(3)}(1)+4\varphi^{(2)}%
(1)\right)  \left(  \partial^{2}/\partial\theta^{2}\right)  C(\theta_{ML,W})\\
\frac{\partial^{2}}{\partial\alpha^{2}}M_{n,3}\left(  \theta_{ML}%
,\alpha\right)  _{\alpha=\theta_{ML,W}}  &  =\varphi^{(2)}(1)\left(
\partial^{2}/\partial\theta^{2}\right)  C(\theta_{ML,W}),
\end{align*}
whence

\bigskip%
\begin{align*}
\frac{\partial}{\partial\alpha}M_{n}\left(  \theta,\alpha\right)
_{\alpha=\theta_{ML,W}}  &  =0\\
\frac{\partial^{2}}{\partial\alpha^{2}}M_{n}\left(  \theta,\alpha\right)
_{\alpha=\theta_{ML,W}}  &  =-\varphi^{(2)}(1)\left(  \partial^{2}%
/\partial\theta^{2}\right)  C(\theta_{ML,W})
\end{align*}
which proves (\ref{alfaMaximise}), and closes the proof.

In contrast with the i.i.d. sampling case minimum divergence estimators in
exponential families under appropriate weighted sampling do not coincide
independently upon the divergence.

\bigskip

\section{Weak behavior of the weighted sampling MLE's}

The distribution of the estimator is obtained under the sampling scheme which
determines its form. Hence under the weighted sampling one. So the observed
sample $x_{1},...,x_{n}$ is considered non random, and is assumed to satisfy%
\[
\lim_{n\rightarrow\infty}\frac{1}{n}\sum_{i=1}^{n}\delta_{x_{i}}=P_{\theta
_{T}}%
\]
and randomness is due to the set of i.i.d. weights $W_{1},...,W_{n}.$

All those estimators can be written as approximate linear functionals of the
weighted empirical measure $P_{n}^{W}$. Therefore all the proofs in
\cite{BK2009} can be adapted to the present estimators.\ Even the asymptotic
variances of the estimators are the same, and subsequently, Wilk's tests ,
confidence areas, minimum sample sizes certifying a given asympptotic power,
etc, remain unchanged.\ The only arguments to be noted are the following: All
arguments pertaining to laws of large numbers for functionals of the empirical
measure carry over to the present setting, conditionally on the observations
$x_{1},...,x_{n}$ . Indeed consider a statistics
\[
U_{n}:=\frac{1}{n}\sum_{i=1}^{n}W_{i}f(x_{i})
\]
where the function $f$ satisfies%
\[
\lim_{n\rightarrow\infty}\frac{1}{n}\sum_{i=1}^{n}f(x_{i})=\mu_{1,f}<\infty
\]
and
\[
\lim_{n\rightarrow\infty}\frac{1}{n}\sum_{i=1}^{n}{}f^{2}(x_{i})=\mu
_{2,f}<\infty.
\]
Then clearly%
\[
\lim_{n\rightarrow\infty}EU_{n}=\mu_{1,f}%
\]
and
\[
\lim_{n\rightarrow\infty}VarU_{n}=\mu_{2,f}-\left(  \mu_{1,f}\right)  ^{2}.
\]
Weak behavior of the estimates follow also from similar arguments: Consider
for example the statistics%
\[
T_{n}:=\sqrt{n}\left(  U_{n}-\mu_{1,f}\right)  /\sqrt{\mu_{2,f}-\left(
\mu_{1,f}\right)  ^{2}}.
\]
Using Lindeberg Central limit theorem for triangular arrays , we obtain that
$T_{n}$ is asymptotically standard normal conditionally upon $x_{1},...,x_{n}%
$.$\ $ It follows that the limit distributions of $\widetilde{\phi^{W}}%
(\theta,\theta_{T})$ and of $\widehat{\theta}_{ML,W}$ conditionally on
$x_{1},...,x_{n}$ coincide with those of $\phi_{n}(\theta,\theta_{T})$ and of
$\widehat{\theta}_{n}$ as stated in \cite{BK2009} under the i.i.d. sampling.
Also all results pertaining to tests of hypotheses are similar, as is the
possibility to handle non regular models.

\bigskip

\section{Proof of Lemma \ref{WeightedLDP}}

\label{proof}

\begin{proof} Recall that $B(S)$ denotes the class of all bounded measurable
functions on $S$. Write $B^{\prime}(S)$ as the algebraic dual of $B(S)$. We
equip $B^{\prime}(S)$ with $B(S)$-topology, it is the weakest topology which
makes continuous the following linear functional:
\[
\zeta\mapsto<f,\zeta>:B^{\prime}(S)\rightarrow\mathbb{R},\text{ for all
}f\text{ in }B(S),
\]
where $<f,\zeta>$ denotes the value of $f(\zeta)$. It follows that
$\mathcal{M}(S)$ is included in $B^{\prime}(S)$ and is endowed with the
$\tau_{0}$-topology induced by $B(S)$. Construct the projection:
$p_{f_{1},...,f_{m}}:B^{\prime}(S)\rightarrow\mathbb{R}^{m},m\in\mathbb{Z}%
_{+}$, namely, $p_{f_{1},...,f_{m}}(\zeta)=(<f_{1},\zeta>,...,<f_{m}%
,\zeta>),f_{1},...f_{m}\in B(S)$. Then for $p_{f_{1},...,f_{m}}(P_{n,\theta
}^{W})=(<f_{1},P_{n,\theta}^{W}>,...,<f_{m},P_{n,\theta}^{W}>)$ we define the
corresponding limit logarithm moment generating function as follows
\begin{align*}
h(t) &  :=\lim_{n\rightarrow\infty}\frac{1}{n}\log\mathbb{E}(\exp
(n<t,Y_{m}>))=\lim_{n\rightarrow\infty}\frac{1}{n}\log\mathbb{E}(\exp
(\sum_{j=1}^{m}<t_{j}f_{j},\sum_{i=1}^{n}W_{i}\delta_{x_{i}}>))\\
&  =\lim_{n\rightarrow\infty}\frac{1}{n}\sum_{i=1}^{n}\log\mathbb{E}%
\exp\left(  \sum_{j=1}^{m}t_{j}f_{j}(x_{i})W_{i}\right)  =\int\left(
\sum_{j=1}^{m}M(t_{j}f_{j})\right)  dP_{\theta}%
\end{align*}
where $t=(t_{1},...,t_{m})\in\mathbb{R}^{m}$ and $Y_{m}=(<f_{1},P_{n,\theta
}^{W}>,...,<f_{m},P_{n,\theta}^{W}>)$. The function $h(t)$ is finite since
$f\in B(S)$. $M(f)$ is Gateaux-differentiable since the function $s\rightarrow
M(f+sg)$ is differentiable at $s=0$ for any $f,g\in B(S)$
\[
\frac{d}{ds}M(f+sg)|_{s=0}=\frac{\int{ge^{f}dP_{W}}}{\int e^{f}dP_{W}},
\]
where $P_{W}$ is the law of $W$. Further, the Gateaux-differentiability of
$M(f)$ together with the interchange of integration and differentiation
justified by dominated convergence theorem show that $h(t)$ is also
Gateaux-differentiable in $t=(t_{1},...,t_{m})$. Hence by the Gartner-Ellis
Theorem (see e.g. Theorem $2.3.6$ of \cite{Dembo}), $p_{f_{1},...,f_{m}%
}(P_{n,\theta}^{W})$ satisfies the LDP in $\mathbb{R}^{m}$ with the good rate
function
\begin{align}
\Phi_{f_{1},...,f_{m}}(<f_{1},\zeta>,...,<f_{m},\zeta>) &  =\sup
_{t_{1},...,t_{m}\in\mathbb{R}}\Big\{\sum_{i=1}^{m}t_{i}<f_{i},\zeta>-\int
M\left(  \sum_{i=1}^{m}t_{i}f_{i}\right)  dP_{\theta}\Big\}\nonumber\\
&  \leq\sup_{f\in B(S)}\Phi_{f}(<f,\zeta>):=\phi^{W}\left(  \zeta,P_{\theta
}\right)  .\label{garterEllis}%
\end{align}
Since $m$ is arbitrary positive integer, by Dawson-Gartner's Theorem (see e.g.
Theorem $4.6.1$ of \cite{Dembo}), $P_{n,\theta}^{{W}}$ satisfies the LDP in
$B^{\prime}(S)$ with the good rate function $\phi^{W}\left(  \zeta,P_{\theta
}\right)  $, which is:
\begin{align*}
\phi^{W}\left(  \zeta,P_{\theta}\right)   &  =\sup_{f\in B(S)}\Phi
_{f}(<f,\zeta>)=\sup_{f\in B(S)}\Big\{\int_{S}f(x)\zeta(dx)-\int%
_{S}M(f)P_{\theta}(dx)\Big\}\\
&  =\int_{S}M^{\ast}\left(  \frac{d\zeta}{dP_{\theta}}\right)  dP_{\theta},
\end{align*}
note that $B^{\prime}(S)$ is endowed with the $\tau_{0}$-topology, the proof
of last equality is given below. Here we always assume $\zeta$ is absolutely
continuous with respect to $P_{\theta}$, otherwise $\phi^{W}(\zeta,P_{\theta
})=\infty$. Consider $\mathcal{M}(S)\subset B^{\prime}(S)$, and set $\phi
^{W}(\zeta,P_{\theta})=\infty$ when $\zeta\notin\mathcal{M}(S)$. Hence
$P_{n,\theta}^{W}$ satisfies the LDP in $\mathcal{M}(S)$ with the rate
function $\phi^{W}\left(  \zeta,P_{\theta}\right)  $, for $\zeta\in
\mathcal{M}(S)$. As mentioned before, $\mathcal{M}(S)$ is endowed with the
topology induced by $B^{\prime}(S)$, namely the $\tau_{0}$-topology.
Nown we give another representation of the rate function $\phi^{W}%
(\zeta,P_{\theta})$. We have:
\begin{align*}
&  \sup_{\zeta\in\mathcal{M}(S)}\Big\{\int_{S}f(x)\zeta(dx)-\int_{S}M^{\ast
}\left(  \frac{d\zeta}{dP_{\theta}}\right)  dP_{\theta}\Big\}\\
&  =\sup_{\zeta\in\mathcal{M}(S)}\Big\{\int_{S}\left(  \int_{S}fd\zeta
-M^{\ast}\left(  \frac{d\zeta}{dP_{\theta}}\right)  \right)  dP_{\theta
}\Big\}\leq\int_{S}M(f)dP_{\theta},
\end{align*}
where the inequality holds from the duality lemma and when $d\zeta
=(dP_{\theta})M^{\prime}(f)$ the equality holds. Using once again the duality
lemma, we obtain the following identity:
\[
\int_{S}M^{\ast}\left(  \frac{d\zeta}{dP_{\theta}}\right)  dP_{\theta}%
=\sup_{\zeta\in\mathcal{M}(S)}\Big\{\int_{S}f(x)\zeta(dx)-\int_{S}%
M(f)dP_{\theta}\Big\}=\phi^{W}(\zeta,P_{\theta}).
\]
The convexity of the rate function $\zeta\rightarrow\phi^{W}\left(
\zeta,P_{\theta}\right)  $ holds from Theorem $7.2.3$ of \cite{Dembo} where
they show the convexity of $\phi^{W}\left(  \zeta,P_{\theta}\right)  $ on
$\mathcal{M}(S)$ endowed with $B(S)$-topology. Hence this is also applied to
$\tau_{0}$-topology which is induced by $B(S)$-topology.
This completes the proof of the lemma.
\end{proof}

\begin{rem}
By the classical Gartner-Ellis Theorem, in $(\ref{garterEllis})$, the
essential smoothness of $h(t)$ is needed for $\Phi_{f_{1},...,f_{m}}$ to be a
\lq\lq good rate function". But on a locally convex Hausdorff topological
vector space, the essential smoothness of $h(t)$ can be reduced to Gateaux
differentiability; see Corollary $4.6.14$ (page $167$) and the proof Theorem
$6.2.10$ (page $265$) of $\cite{Dembo}$.
\end{rem}

\begin{rem}
Since $\Phi_{f_{1},...,f_{m}}(<f_{1},\zeta>,...,<f_{m},\zeta>)$ is a good rate
function in $\mathbb{R}^{m}$, its level sets $\Phi_{f_{1},...,f_{m}}%
^{-1}(\alpha)=\{(y_{1},...,y_{m})\in\mathbb{R}^{m}:\Phi_{f_{1},...,f_{m}%
}(y_{1},...,y_{m})\leq\alpha\}$ are compact, $\ $for all $\alpha$ in
$[0,\infty)$. Denote the projective limit of $\Phi_{f1,...,f_{m}}^{-1}%
(\alpha)$ by $\Phi_{f}^{-1}(\alpha)=\underleftarrow{\lim}\Phi_{f1,...,f_{m}%
}^{-1}(\alpha)$. According to Tychonoff's theorem, the projective limit
$\Phi_{f}^{-1}(\alpha)$ of the compact set $\Phi_{f1,...,f_{m}}^{-1}(\alpha)$
is still compact, so $\phi^{W}\left(  \zeta,P_{\theta}\right)  = \sup_{f\in
B(S)}\Phi_{f}(<f,\zeta>)$ is also a good rate function in $(\mathcal{M}%
(S),\mathcal{B}(\mathcal{M}))$.
\end{rem}

\end{document}